\documentclass[sigconf,10pt]{acmart}
\renewcommand\footnotetextcopyrightpermission[1]{} 
\usepackage{graphicx}
\usepackage{balance}  
\usepackage{enumerate}
\usepackage{enumerate}
\usepackage{breqn}
\usepackage{calligra}
\usepackage[ruled,vlined,linesnumbered]{algorithm2e}

\usepackage{amsthm}
\newtheorem{theorem}{Theorem}

\usepackage{tikz}
\usetikzlibrary{chains,positioning,matrix,arrows,scopes,shapes,fit,calc}

\setcopyright{rightsretained}

\copyrightyear{2021}

\acmArticle{4}
\acmPrice{15.00}

\begin{document}
\title{Secure Hypersphere Range Query on Encrypted Data}

\author{Gagandeep Singh} \authornote{This work was done when the author was with IBM Research, India}
\affiliation{ \institution{Oblivious AI}}
\email{gagan.jld@gmail.com}

\author{Akshar Kaul}
\affiliation{ \institution{IBM Research, India}}
\email{akshar.kaul@in.ibm.com}
%
%
%

%
%



\begin{abstract}
	Spatial queries like range queries, nearest neighbor, circular range queries etc. are the most widely used queries in the location-based applications. Building secure and efficient solutions for these queries in the cloud computing framework is critical and has been an area of active research. 
This paper focuses on the problem of \emph{Secure Circular Range Queries (SCRQ)}, where client submits an encrypted query (consisting of a center point and radius of the circle) and the cloud (storing encrypted data points) has to return the points lying inside the circle. 
The existing solutions for this problem suffer from various disadvantages such as high  processing time which is proportional to square of the query radius, query generation phase which is directly proportional to the number of points covered by the query etc.

This paper presents solution for the above problem which is much more efficient than the existing solutions. Three protocols are proposed with varying characteristics. It is shown that all the three protocols are secure. The proposed protocols can be extended to multiple dimensions and thus are able to handle \emph{Secure Hypersphere Range Queries (SHRQ)} as well.
Internally the proposed protocols use pairing-based cryptography and a concept of lookup table. To enable the efficient use of limited size lookup table, a new storage scheme is presented. The proposed storage scheme enables the protocols to handle query with much larger radius values. 
Using the \emph{SHRQ} protocols, we also propose a mechanism to answer the \emph{Secure range Queries.} 
Extensive performance evaluation has been done to evaluate the efficiency of the proposed protocols.\\





\end{abstract}

\maketitle

\section{Introduction}
Spatial data is widely used by various types of applications such as location-based services, computer aided design, geographic information system etc. The individual data point in such databases is represented as a point in the euclidean space. The applications query these spatial databases via various  types of geometric queries. One of the most widely used geometric query is the circular range query. The circular range query is defined by a center and a radius. The goal of the circular range query is to find all those data points which lie inside the circle defined by the circular range query. When the data has more than two dimensions, then it becomes a hypersphere query. 

The hypersphere queries are widely used to provide various services such as finding the point of interest within a certain distance from the user. For example, finding all the restaurants which are within 5 miles from the user's location.

In recent years with the emergence of cloud computing, various enterprises are outsourcing their data and computation needs to third party cloud service providers. Use of cloud services make economic sense for these organizations. However, it brings its  share of problems as well. One of the most prominent amongst them being security and privacy of data. Enterprises have to trust cloud service provider with their data which includes sensitive information as well. 

The most basic way to protect the data is to encrypt it using a secure encryption scheme such as AES, before transferring to the cloud. By doing so,  usability of the data outsourced to the cloud is reduced as no query can now be executed at the cloud server. For each query, the user has to download the data, decrypt it and run the query which will not be acceptable.

To solve this dual problem of securing the data as well as maintaining its usability, searchable encryption schemes have been proposed. These encryption schemes encrypt the data in such a way that one or more desired properties of data are preserved. For example, deterministic encryption of data preserves equality. 
In the context of searchable encryption schemes, problem of handling geometric queries over encrypted data has received very less attention. The existing solutions consists of two paper from Wang \emph{et al.} \cite{wang2015circular,wang2016geometric}.

In \cite{wang2015circular}, solution for secure circular range queries  over encrypted data was proposed. It used the existing predicate encryption scheme as base, and  developed the protocol to test ``if the query point lies on the circle or not''. For circular range query, this protocol is called repeatedly for all the concentric circles centered at query point and having  radius $\leq$ query radius ($R$). 
The problem with this approach is that the number of such concentric circles is $\mathcal{O}(R^2)$. Thus degrading the performance by a factor of $\mathcal{O}(R^2)$ and  thus depicting the non-scalability of the solution.

 In\cite{wang2016geometric},  another solution for handling any geometric query over encrypted spatial data was proposed. The  solution used a combination of bloom filters and predicate encryption scheme. In the scheme, for each data point a boom filter is created and stored at cloud and for geometric query, all the data points that lie within the query area are enumerated and inserted into a query bloom filter. 
 The cloud operation includes the comparison of the query bloom filter with the data point bloom filter
  and if the data bloom filter is fully captured inside the query bloom filter, it is added to the output. 
The operation is performed securely by using predicate encryption scheme which encrypts the bloom filters and converts the  bloom filter comparison to secure inner product computation and comparison. 
 The  solution suffers from false positives due to the use of bloom filter and since all the  points within query area have to be enumerated and inserted into query bloom filter, the solution is not scalable.
 
 Considering the shortcomings of two papers, in this paper, we propose a novel \emph{SHRQ} protocols for handling hypersphere range queries over encrypted spatial data. We  design a \emph{Component Encryption Scheme(CES)} which allows evaluation of degree two polynomials directly over ciphertext i.e. the scheme can be used to perform a  dot product between  encrypted data  and query point. 
 At the core of the  protocols, \emph{CES} scheme is used for secure evaluation of result set. For the scalability of solution w.r.t. high radius values, we present the new storage idea of \emph{Coarse Grained Granular Storage}.
  And unlike the existing solutions, the performance of the proposed scheme is found to be independent of query area size. 

We also present a solution to the problem of \emph{Secure Range Query}(SRQ) which is based on our solution for  \emph{SHRQ}. We present the remodeling of the underlying data points and query points that helps in simultaneous support of \emph{SRQ} queries along with  \emph{SHRQ} queries over the encrypted data. With the remodeling, we show that server protocol and servers data imprints of  \emph{SHRQ} and  \emph{SRQ}   queries are similar and therefore server can't even differentiate between the  types of queries, thus making the solution type oblivious.

The security of our solution is presented in standard attack models: \emph{Known Plaintext Attack(KPA)} model and \emph{Ciphertext Only Attack(COA)} model. We claim our system to be secure in these attack models. In each of these models, we  quantify the information that is required by the adversary to decrypt the whole database. 
To summarize, our main contributions are listed below:
\begin{enumerate}
	\item We present the first computationally efficient solution to \emph{Secure Hypersphere Range Query} in a single cloud server based system.
	\item The proposed system can also be used to answer \emph{Secure Range Queries}. To the best of our knowledge, this is the first system that answers the secure range query along with a secure location based query.
	\item We prove our system to be secure under the standard attack models: \emph{KPA} and \emph{COA}. 
\end{enumerate}
\section{Problem Framework}
In this section, we formulate the \emph{SHRQ} problem. We present details of the system setting describing various entities, solution objective and the adversarial model used for security analysis. We also present a brief description of the operations carried out in \emph{Hyphersphere Range queries (HRQ)} in plain text setting.

\subsection{System Setting}
\begin{figure}
	\includegraphics[width=\linewidth]{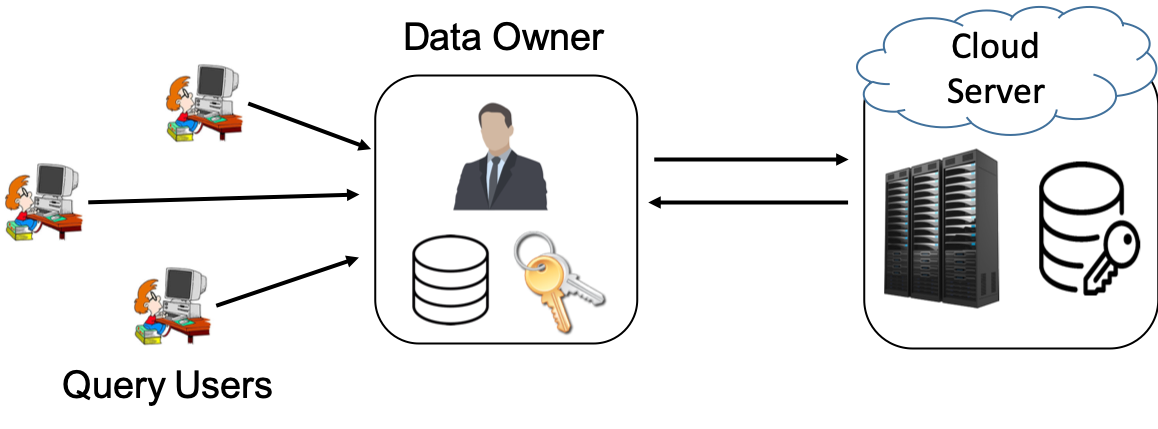}
	\caption{System Entities}
	\label{fig:entities}
\end{figure}

Entities present in a typical outsourced setting are shown in Fig. \ref{fig:entities}.
\begin{enumerate}
	\item \textbf{Cloud Server (CS)}: Cloud Server is a third-party service provider that provides storage and computation resources to its clients. It takes care of all the administration and management of computation hardware. And clients pay him for using its resources. 
	Examples of such service providers are IBM Cloudant \cite{URLIBM}, Amazon AWS \cite{URLAmazon}, Microsoft Azure \cite{URLMicrosoft} etc.
	
	\item \textbf{Data Owner (DO)}: In cloud computing framework, Data Owner represents the prospective client of Cloud Server, who is having a propriety ownership on some specific data and is willing to use Cloud Server resources for storage and computational requirements. 
	
	Pertaining to \emph{SHRQ} problem, Data Owner's data is represented by a data table $D =\{m^1,m^2 \cdots ,m^{|D|}\}$ where each $m^i$ is a $d$ dimensional data point.
	For using Cloud Server resources  \emph{securely}, Data Owner encrypts the table $D$ to $D' $ 
	and then outsources $D'$ to the Cloud Server.

	
	\item \textbf{Query User (QU)}: These are the authorized users who want to execute the \emph{SHRQ} queries over the  outsourced encrypted data $D'$. For this work, we assume that Query User will share their query in plaintext with Data Owner. Who, on behalf will perform the query execution protocol with Cloud Server and then share  the decrypted results back to Query User. While there are seperate papers where Query User's privacy is considered seperately and queries are also hidden from the Data Owner. The scope of this paper is constrained to developing a  cloud based SHRQ solution considering security from Cloud Server. The query user privacy consderations will require a seperate study over the proposed solution and thus is kept as a future scope of work.
	
	
\end{enumerate}

\subsection{Objective}\label{objective}
The main objective of Secure Hypersphere Range Query (SHRQ) system is to develop a system which enables Hypersphere Range Queries to be answered efficiently over outsourced data in a privacy preserving manner. 
A Hypersphere Range Query consists of a query point $q$ and a radius $r$. The goal is to find all those points in the database that lie inside the hypersphere defined by $q$ and $r$. 

The privacy preservation requires that, the data points as well as the query points should be protected from Cloud Server. This stems from the fact that Cloud Server is not trusted with  data. Also it has to be ensured that Cloud Server does not learn anything about either the data or query during the query processing phase.

One implicit requirement is that the cloud server should be able to compute the result set of SHRQ efficiently and that processing requirements for Data Owner and Query User should be low. 



\subsection{Adversarial Model}

We consider the security of our scheme in \emph{``Honest But Curious (HBC)"} adversarial model. In this adversarial model, adversary is \emph{honest} in performing all the steps of the protocol correctly but  is also \emph{curious} to learn more about the data which is being shared with him. HBC is the most widely used adversarial model in cloud computing setting. This model gives  adversary the freedom to observe the computations being done, and use the same for data inference. Some characteristics of the HBC model are as follows :-
\begin{itemize}
	\item Adversary does not tamper with the stored data.
	\item Adversary doesn't have the domain knowledge and statistics of the data. 
	\item Adversary knows all the details of the protocols being used. The only thing hidden from him is the secret key.
	\item The goal of adversary is to infer the plain text data of Data Owner.
\end{itemize}

We present security proofs for our protocols in Ciphertext Only Attack Model (COA) and Known Plaintext Attack Model (KPA). These model differ in the amount of powers which the adversary has for breaking the protocol.

\begin{itemize}
	\item Ciphertext Only Attack Model (COA)\\
	In this attack model, adversary only has access to the encrypted database $D'$. He can also observe the query processing computation being done over $D'$. The goal of the adversary is to decipher the encrypted database $D'$ to plaintext $D$.
	\item Known Plaintext Attack Model (KPA)\\
	Like COA model, the adversary has access to the encrypted database $D'$ and can observe the query process computations. Additionally, he also has access to few plain texts and their corresponding cipher text values. The adversary has no control over the plain texts for which he gets the cipher text values. The adversary goal is same as COA, i.e. to decipher the encrypted database $D'$.
\end{itemize}


\subsection{HRQ in Plain Text}
Now we will describe the basic computations required for executing \emph{HRQ} over plain text data. The data point in $D$ is assumed to be a $d$ dimensional vector, represented by $[m_1,m_2,...,m_d]$. And the HRQ query consists of a $d$ dimensional query point (represented by $q = [q_1,q_2,...,q_d]$) and a  radius $r$. We use the notation $Q\{q,r\}$ to represent the \emph{HRQ} query. 

HRQ evaluation between a data point $m$ and query $Q$ involves distance computation and comparison with query radius. 
\begin{enumerate}
	\item First the distance between the data point and query center is computed. \\
	$Dist = \sqrt{(m_1-q_1)^2 + (m_2-q_2)^2 + ... + (m_d-q_d)^2}$
	\item The computed distance $Dist$, is then compared with query radius. The data points whose distance from query center is less than radius are added to output. It requires checking of the following condition:\\
	$r -  Dist \geq 0$
\end{enumerate}

The above two step process can be reduced to the following inequality:
\begin{equation*}
	2(m_1.q_1+m_2.q_2+...+m_d.q_d)+r^2 - ||m||^2 - ||q||^2 \geq 0
\end{equation*}

Observe that the above inequality involves the following basic steps :
\begin{enumerate}[(a)]
	\item \textbf{Dot Product:} Dot product between two $d+2$ dimensional components:
	\begin{enumerate}[1.]
		\item \emph{Data Component} : $\{m_1,m_2,...,m_d,1,||m||^2\}$
		\item \emph{Query Component} : $\{2q_1,2q_2,...,2q_d,r^2-||q||^2,-1\}$
	\end{enumerate} 
	\item \textbf{Verification:} Verify if the \emph{Dot Product} is non negative.
\end{enumerate}

In the proposed \emph{SHRQ} scheme, we present the solution of performing the above two operations securely and with computational efficiency.

\section{Background}
In this section, we are going to present the bilinear mapping based background art, that is used to build the \emph{SHRQ} solution.  Specifically, we give brief overview of bilinear groups of composite order in section \ref{CompositeOrder}. In section \ref{BGN} and \ref{modified_bgn}, we first give a background on existing 2DNF homomorphic encryption scheme and then propose the modified encryption scheme with security that is used by our \emph{SHRQ} protocols.

\subsection{Bilinear  Groups of Composite Order } \label{CompositeOrder}
Following, we give a brief overview of the composite order finite groups supporting  bilinear mapping operations, first introduced in \cite{boneh2005evaluating}.  In our scheme, we use \emph{bilinear groups} of composite order, where order is of the size of product of two prime numbers. \\
Let $\mathcal{G}$ be a group generating algorithm that takes a security parameter $1^\lambda$ and outputs a tuple $(q_1,q_2,\mathbb{G,G}_T,e)$ where $q_1,q_2$ are distinct prime numbers, $\mathbb{G,G}_T$ are cyclic groups of order $N = q_1\times q_2$, and a pairing map: $e : \mathbb{G \times G \rightarrow G}_T$, which satisfies following bilinear properties: 
\begin{enumerate}
	\item Bilinear : $\forall x,y \in \mathbb{G}, \forall a,b \in \mathbb{Z}, e(x^a,y^b) = e(x,y)^{ab}$ 
	\item Non-degeneracy : $\exists g \in \mathbb{G}$ s.t. $e(g,g)$ has an order of $N$ in $\mathbb{G}$
	\item Efficiency : Group operations in $\mathbb{G,G}_T$ and computation of pairing map $e$ can be done efficiently.
\end{enumerate}
Since $\mathbb{G}$ is a composite order group, with order of $q_1\times q_2$, there will exist subgroups of order $q_1 \text{ and } q_2$. We denote such  subgroups as  $\mathbb{G}_{q_1}$ and $\mathbb{G}_{q_2}$. We will also use the following \emph{mathematical facts} related to bilinear groups of composite order in the scheme construction :
\begin{itemize}
	\item $\forall x \in \mathbb{G}_{q_1}, \forall y \in \mathbb{G}_{q_2}, e(x,y)=1 $.
	\item $\forall x,y \in \mathbb{G}$ s.t. $x=x_{q_1}x_{q_2}$,  $y=y_{q_1}y_{q_2}$ the mapping $e(x,y) = e(x_{q_1},y_{q_1})e(x_{q_2},y_{q_2})$ holds when $x_{q_1}, y_{q_1}\in \mathbb{G}_{q_1}$ and $x_{q_2}, y_{q_2}\in \mathbb{G}_{q_2}$.
\end{itemize}

\subsection[BGN Cryptosystem]{BGN Encryption Scheme \cite{boneh2005evaluating}}\label{BGN}
This encryption scheme $(\mathcal{E}_{BGN})$ supports 2DNF homomorphic computation over the ciphertext i.e. it supports single multiplication operation and unlimited subsequent additions over encrypted ciphertext. The 2DNF homomorphic property  also helps in \emph{Dot Product} evaluation. 

The Encryption protocols defined under $\mathcal{E}_{BGN}$ are as follows:
\begin{itemize}
	\item KeyGen : Given a security parameter $1^\lambda$, $\mathcal{G}$ generates a tuple $(q_1,q_2,\mathbb{G,G}_T,e)$. 
	Generate two random generators $g,u$ from group $\mathbb{G}$, and set $h = u^{q_2}$, here $h$ would be a generator of subgroup of $\mathbb{G}_{q_1}$. Public Key, $\mathcal{PK} = \{{N,\mathbb{G,G}_T,e,g,h}\}$ and Secret Key, $\mathcal{SK} = \{q_1\}$.
	\item Enc$(\mathcal{PK},m)$ : Pick a random element $r \in \mathbb{Z}_N$ and compute ciphertext $C = g^mh^r$. 
	\item Dec$(\mathcal{SK},C)$ : The ciphertext $C$ could belong to $\mathbb{G}$ or $\mathbb{G}_T$. For decryption, DO first computes modular exponentiation of ciphertext with $\mathcal{SK}$, i.e. $(C)^{q_1} = (g^mh^r)^{q_1} = (g^m)^{q_1}$ and then have to compute the costly discrete log of $(g^m)^{q_1}$ w.r.t. $q_1$. 
\end{itemize}
Computation of plaintext message from decryption protocol value, Dec$(\mathcal{SK},C)$  can be done by following ways:
\begin{enumerate}
	\item Pollard's lambda\cite{boneh2005evaluating}: The complexity of this method is computationally high, of the order of $\mathcal{O}(\sqrt{T})$ where $T$ is the message space, which is directly proportional to message space
	\item Maintain a \emph{lookup table}, mapping plaintext messages to decryption protocol value, Dec$(\mathcal{SK},C)$. This table is pre-computed and stored in the database for faster decryption. Since the table is stored in database, it is bounded by the storage space and thus can store limited set of values from the message space. 
\end{enumerate}
Both the above points suggest that either the scheme $\mathcal{E}_{BGN}$ works efficiently when the  message space is small, or the underlying system should use precomputed \emph{lookup table}. The challenge of using the lookup table lies in the fact of limited storage space. Thus requiring a design decision on what values to store or a strategy for optimum use of table for values that are not in table.

The homomorphic properties of the scheme can be observed as follows:
\begin{enumerate}
	\item Multiplication: Given the cipher text $C_1,C_2 \in \mathbb{G}$ w.r.t. messages $m_1,m_2$. One time multiplication can be supported using bilinear mapping as following
	\begin{align*}
	C &= e(C_1,C_2)=e(g^{m_1}h^{r_1},g^{m_2}h^{r_2})\\
	&=e(g,g)^{m_1m_2}.e(g,h)^{m_1r_2+m_2r_1+\alpha q_2r_1r_2}\\
	&=g_1^{m_1m_2}h_1^{r'} \in \mathbb{G}_T
	\end{align*}
	\item Addition : Clearly it can be observed that the cipher texts $C_1,C_2 \in \{ \mathbb{G}\text{ or }\mathbb{G}_T\}$ supports addition by using cyclic multiplicative operation in both the groups.
\end{enumerate}

\subsection{Component Encryption Scheme(CES)}\label{modified_bgn}


As mentioned in section \ref{objective}, our system involves computation w.r.t. \emph{Data Component} and \emph{Query Component}.
Therefore at first, every $d$ dimensional data (\emph{data point} or \emph{query point}) is first converted to $d+2$ dimensional \emph{Component} (\emph{Data Component} or \emph{Query Component}).

For encryption of \emph{Data Component} and \emph{Query Component} as whole, we propose the \emph{Component Encryption Scheme} $(\mathcal{E}_{CES})$ which maintains the necessary computation support over 2DNF formula over finite dimensional vector.

Following are the various protocols involved in $\mathcal{E}_{CES}$:
\begin{itemize}
	\item KeyGen : Call the group generating algorithm with security parameter $1^\lambda$ which generates the parameters $\{q_1,q_2,\mathbb{G,G}_T,e\}$. Here, $q_1,q_2$ are prime numbers and $N=q_1\times q_2$ is the order of composite groups $\mathbb{G,G}_T$. Further we generate following parameters :
	\begin{enumerate}[a.]
		\item Like $\mathcal{E}_{BGN}$ Cryptosystem, generate  the parameters $g,h$ and then set another parameter $s = g^{q_1}$. Both $s$ and $h$, are elements of subgroups $\mathbb{G}_{q_1}$ and $\mathbb{G}_{q_2}$ respectively of group $\mathbb{G}$. 
		\item Generate two $(d+2)$ size random vectors $A,B$ s.t. dot product, $A.B^T = \text{\LARGE\calligra{r}}.q_1$. Here,{ \LARGE\calligra r} is a random integer in $Z_N$. 
	\end{enumerate}
	Parameters $\{N,\mathbb{G,G}_T,e\}$ are shared with cloud and all other information is kept secret.
	\item TupleEncryption: This is a randomized protocol. Select a random integer$\text{ \LARGE\calligra r}_m \in Z_N$. And encrypt each attribute of the data component as:
	\begin{equation*}
	m'_i =s^{m_i}h^{\text{\large\calligra r}_mA_i} \text{ ,where  } i \in [0,n+2]
	\end{equation*}
	\item QueryEncryption: This procedure generates encryption of query component. DO choses three random integers $\alpha,\beta, \text{\LARGE\calligra r}_q \in Z_N$ and encrypts the Query Component as below:
	\begin{equation*}
	q'_i =
	\begin{cases}
	s^{(q_i + \beta).\alpha}h^{\text{\large\calligra r}_qB_i} & \text{if }i=n+1\\
	s^{q_i. \alpha}h^{\text{\large\calligra r}_qB_i}&\text{otherwise}\\
	\end{cases}
	\end{equation*}
	%
	
\end{itemize}
\subsubsection{Security}~\\
The proposed scheme $\mathcal{E}_{CES}$  varies from the $\mathcal{E}_{BGN}$ encryption scheme by following ways:
\begin{enumerate}
	\item[A.] \emph{Parameter sharing is more restrictive. While BGN shares $\{N,\mathbb{G,G}_T,e,g,h\}$ as public key, the proposed system only shares $\{N,\mathbb{G,G}_T,e\}$ as computation parameters.}
	\item[B.] \emph{The cipher text format of $\mathcal{E}_{BGN}$ system is of form $g^mh^r$, while $\mathcal{E}_{CES}$  encrypts data and query attributes seperately with both having ciphertext format as $s^{x_i}h^{r.Y_i}$.}
\end{enumerate}
From security stand point w.r.t. point A, it can be generalized that \emph{ 
the security of encryption scheme with more restriction on public parameters can not be less than the encryption scheme with more parameters.}

For point B, we prove the security of $\mathcal{E}_{CES}$ using  semantic security assumption of $\mathcal{E}_{BGN}$ as base. We refer reader to section 3.2 of \cite{boneh2005evaluating} for detailed security understanding of $\mathcal{E}_{BGN}$. 
In brief, the semantic security claim of  $\mathcal{E}_{BGN}$ states that : \emph{Given a ciphertext $m_b'$ of one of the messages, $m_0$ or $m_1$, it is computationally hard for any polynomial time adversary to identify that if $m_b'$ is encryption of $m_0$ or $m_1$}.

We prove that if $\mathcal{E}_{CES}$ is not secure then so would be $\mathcal{E}_{BGN}$, which goes against the base assumption of semantic securty of $\mathcal{E}_{BGN}$. And since the base assumption is known to be \emph{True}, the proposition that \emph{``$\mathcal{E}_{CES}$ is not secure''} should be \emph{False}.

For proof, we assume another transitive system $(\mathcal{E}_{TRANS})$ that produces cipher text of form $s^m.h^r$ w.r.t.  $\mathcal{E}_{BGN}$.
Security is proven by following implication theorems. 
\begin{equation*}
Security(\mathcal{E}_{BGN}) \implies Security(\mathcal{E}_{TRANS}) \implies Security(\mathcal{E}_{CES})
\end{equation*}
\begin{theorem}\label{Theorem1}
	The security of encryption scheme $\mathcal{E}_{BGN}$ implies the security of encryption scheme  $\mathcal{E}_{TRANS}$
	
	$Security(\mathcal{E}_{BGN}) \implies Security(\mathcal{E}_{TRANS})$
\end{theorem}
\begin{proof}
	Assuming that scheme is not secure, then there exist a polynomial time adversarial algorithm $\mathcal{B}$ that breaks the security of   $\mathcal{E}_{TRANS}$. Using $\mathcal{B}$ we can construct adversarial algorithm $\mathcal{A}$ that breaks the security claim of challenger  $\mathcal{C}$ w.r.t.  $\mathcal{E}_{BGN}$. Initially $\mathcal{C}$  gives the public key of $\mathcal{E}_{BGN}$ to $\mathcal{A}$ as input, $\mathcal{A}$ works as follows:
	\begin{enumerate}
		\item $\mathcal{A}$ shares computation parameters with $\mathcal{B}$ .
		\item $\mathcal{B}$ outputs two messages $m_0,m_1$ to $\mathcal{A}$ as challenge text which are sent as it is to $\mathcal{C}$. $\mathcal{C}$ responds to $\mathcal{A}$ with the encryption of $m_b$ s.t. \emph{CT = }$g^{m_b}h^r$, A further sends the \emph{CT} to $\mathcal{B}$.
		\item $\mathcal{B}$ outputs its guess $b'$.
	\end{enumerate}
	\emph{Ciphertext Correctness:} While $\mathcal{C}$ sends the cipher text of format $g^mh^r$ to $\mathcal{B}$, $\mathcal{B}$ expects the format to be of form $s^mh^r$. Following we show that both represents the same:
	\begin{align*}
	\text{\emph{CT}} &= g^mh^r\\
	&= (s.h')^mh^r
	\end{align*} 
	Since $g \in \mathbb{G}$ (composite order group), it can be expressed in terms of  some specific elements of  underlying subgroups i.e. $s \in \mathbb{G}_{q_1}$ and  $h' \in \mathbb{G}_{q_2}$. As also, $h \in \mathbb{G}_{q_2}$, we can represent $h' = h^\gamma$ for some $\gamma \in (1,\cdots,q_1)$. CT  now becomes:
	\begin{align*}
	\text{\emph{CT}} &= s^mh^{r+m\gamma}\\
	&=s^mh^{r'} 
	\end{align*} 
	
	As $\mathcal{B}$ wins the semantic security game against $\mathcal{E}_{TRANS}$, the constructed adversary $\mathcal{A}$ will win security game against $\mathcal{E}_{BGN}$ which goes against the $\mathcal{E}_{BGN}$ security. \hfill $\qedhere$
\end{proof}

\begin{theorem}
	The security of encryption scheme $\mathcal{E}_{TRANS}$ implies the security of encryption scheme  $\mathcal{E}_{CES}$
	
	$Security(\mathcal{E}_{TRANS}) \implies Security(\mathcal{E}_{CES})$
\end{theorem}
\begin{proof}
	The proof follows the similar flow of assuming the existence of polynomial adversarial algorithm $\mathcal{B}$ that breaks the security of $\mathcal{E}_{CES}$ where \emph{CT} is of format $s^m.h^{r.A}$. Using this adversary we can construct another $\mathcal{A}$ that breaks the security claim of challenger $\mathcal{C}$ w.r.t. $\mathcal{E}_{TRANS}$ which goes against the theorem \ref{Theorem1} and thus against security of the $\mathcal{E}_{BGN}$. The game proceeds as below:
	\begin{enumerate}
		\item $\mathcal{C}$ shares computation parameters with $\mathcal{B}$, who further shares the same with $\mathcal{A}$ .
		\item $\mathcal{B}$ outputs two messages $m_0,m_1$ to $\mathcal{A}$ as challenge text. $\mathcal{A}$ in turn sends two messages $m_0.A^{-1},m_1.A^{-1}$ to $\mathcal{C}$. $\mathcal{C}$ responds to $\mathcal{A}$ with the encryption of $m_b.A^{-1}$ s.t. \emph{CT = }$s^{m_b.A^{-1}}h^r$. $\mathcal{A}$ further process \emph{CT}, s.t. \emph{CT'} = $(\text{\emph{CT}})^A$  and sends \emph{CT'} to $\mathcal{B}$.
		\item $\mathcal{B}$ outputs its guess $b'$.
	\end{enumerate}
	\emph{Ciphertext Correctness:} It can be observed that \emph{CT'} received by $\mathcal{B}$ is of correct format as expected. 
	\begin{align*}
	\text{\emph{CT'}} &=\text{\emph{CT}}^A\\
	&= s^{m_b}h^{r.A}
	\end{align*} 
	As $\mathcal{B}$ wins the security game against $\mathcal{E}_{CES}$, the constructed $\mathcal{A}$ will win the security game against $\mathcal{E}_{TRANS}$ which goes against the  $\mathcal{E}_{TRANS}$ security proven in  theorem \ref{Theorem1} and hence against  $\mathcal{E}_{BGN}$. \hfill  $\qedhere$
\end{proof}
\subsubsection{Useful Function}
%
\begin{enumerate}[A.]
	\item \textbf{Compute(m',q'):} This procedure is called by CS during the  Query phase of \emph{SHRQ} protocol. It takes $\mathcal{E}_{CES}$ encrypted data component  and  query component vectors as inputs
	and computes deterministic dot product value of the two vectors. Shown as below:\\
	1. \textbf{Homomorphic Multiplication:} First, for each attribute in $m'$ and $q'$ a bilinear mapping is computed as follow
	\begin{align*}
	e(m'_i,q'_i) =
	\begin{cases}
	e(s^{m_i}h^{\text{\large\calligra r}_mA_i},s^{(q_i + \beta). \alpha}h^{\text{\large\calligra r}_qB_i}) & \text{if }i=n+1\\
	e(s^{m_i}h^{\text{\large\calligra r}_mA_i},s^{q_i \alpha}h^{\text{\large\calligra r}_qB_i})&\text{otherwise}\\
	\end{cases}
	\\=
	\begin{cases}
	e(s,s)^{m_i(q_i + \beta).\alpha}e(h,h)^{\text{\large\calligra r}_m\text{\large\calligra r}_qA_iB_i} & \text{if }i=n+1\\
	e(s,s)^{m_iq_i \alpha}e(h,h)^{\text{\large\calligra r}_m\text{\large\calligra r}_qA_iB_i})&\text{otherwise}\\
	\end{cases}
	\end{align*}
	2. \textbf{Homomorphic Addition:}Now, compute the modular multiplication operations in group $\mathbb{G}_T$ of $n+2$ values. The value $T$ thus computed will be:
	\begin{align*}
	T &= e(s,s)^{(m.q+\beta)\alpha}e(h,h)^{\text{\large\calligra r}_m\text{\large\calligra r}_qAB}\\
	&= e(s,s)^{(m.q+\beta)\alpha}&\text{as } e(h,h)^{\text{\large\calligra r}_m\text{\large\calligra r}_qAB}=1
	\end{align*}
	The computed deterministic exponential modular dot product value is protected by the hardness assumption of \emph{Discrete Logarithm Problem}. Under  \emph{Discrete Logarithm Problem} assumption, given the exponentiation value it is hard to compute the exponent w.r.t. base in big cyclic group.
	
	\item \textbf{CreateLookupTable($v$):} This function is called by DO for a Lookup table($\mathcal{L}$) construction. Lookup table is used by CS during the Query phase of $SHRQ$ protocol. Before construction of table, both CS and DO agree upon \emph{secure hash function ($\mathcal{H}$)} and thereafter DO constructs a table $\mathcal{L}$ s.t. $\mathcal{L} = \{\mathcal{H}(e(s,s)^{(i+\beta)\alpha}),\text{ where }i=[0,v]\}$
	
\end{enumerate}

\section{SHRQ Protocols}\label{protocols}
In this section, we will present new \emph{SHRQ} protocols which can be used for Secure Hypersphere Range Query execution. 
In section \ref{protocol1}, we present the protocol $SHRQ_T$ where we used the  \emph{lookup table} idea along with  $\mathcal{E}_{CES}$ for query execution. 
In section \ref{sec:cggs}, we introduce the idea of \emph{Coarse Grained Granular Storage} which is used in protocol $SHRQ_C$ presented in section \ref{protocol2}. $SHRQ_C$ improves maximum supported query radius compared to \emph{$SHRQ_T$} but introduces \emph{false positives} to the solution. 
In section \ref{protocol3}, we present the protocol $SHRQ_L$ which overcome the limitation of $SHRQ_C$ with \emph{no false positives} in the solution set. 

Each of these protocols have two phases:
\begin{itemize}

	\item \emph{Setup Phase :}
	In this phase, DO generates the encryption keys and uses them to encrypt the database and generate the Lookup Table. The encrypted database and Lookup Table is then uploaded to the CS.
	 
	\item \emph{Query Phase :}
	This phase represents the scenario of QU querying the encrypted database stored at CS. It involves QU encrypting the query and then sending the same to CS. CS then executes the encrypted query to get the data points which satisfy the Hypersphere Range Query. The data points in the result set are then given to the QU who decrypts them to get the plain text result set.
	
\end{itemize}


\subsection{Protocol SHRQ\texorpdfstring{\textsubscript{T}}{\texttwoinferior}}\label{protocol1}
This section describes the protocol $SHRQ_T$ which uses ${E}_{CES}$ to build a protocol for Secure Hypersphere Range Query execution. 

\subsubsection{Setup Phase}\label{prot1-setupphase}
This phase is executed by the DO  to encrypt and upload the database to CS. 
\begin{enumerate}
	\item DO invokes Keygen of $\mathcal{E}_{CES}$. 
	\begin{itemize}
		\item Parameters $\{N, \mathbb{G, G}_{T},e\}$ are shared with the CS.
		\item Other parameters are kept secret by the DO.
	\end{itemize}
	Also, DO and CS jointly agree upon a \emph{secure hash function} $\mathcal{H}$ used in construction of $\mathcal{L}$
	\item DO encrypts  data point in the database as follows :
	\begin{itemize}
		\item Compute $d+2$ dimensional data component representation $m$ w.r.t. data point and assigns a unique id $(\mathring{i})$ to $m$.
		
		\item  Use \emph{TupleEncryption} function of $\mathcal{E}_{CES}$ to encrypt  $m$. The output is encrypted data point $m'$.
		
		\item Store ($\mathring{i}$, $m'$) at the CS in database \emph{db-query}.
		
		\item Encrypt the data point with a separate encryption scheme like \emph{AES}. Let \~{m} represents the encryption of data point with \emph{AES}. The \emph{AES} encrypted data point is used for faster decryption of the data points that are part of the result set.
		
		\item Store ($\mathring{i}$, \~{m}) at the CS in database \emph{db-store}.
	\end{itemize}

	\item DO calls the procedure \emph{CreateLookupTable}$(v)$ to generate lookup table ($\mathcal{L}$) of size $v$. DO decides the size based on his requirement and system configuration of CS. DO uploads the generated $\mathcal{L}$ to CS.
\end{enumerate}

\subsubsection{Query Phase}\label{prot1-queryphase}
This phase is initiated by the QU to execute a hyper sphere range query, $Q\{q,r\}$ over the encrypted database. 
\begin{enumerate}
	\item If $r > \sqrt{v}$ then the query is not supported by the current instantiation. No further processing is done.
	
	\item QU encrypts the query component w.r.t.  $Q\{q,r\}$ using \emph{QueryEncryption} to get encrypted query $q'$. QU then sends $q'$ to the CS for further processing.
	

	\item CS starts with an empty result set $Res$. For each of the encrypted data point $m'$ in \emph{db-query}, cloud server performs the following steps:
	\begin{enumerate}[i]
		\item Call \emph{Compute} function with  $m'$ and  $q'$. Let $T = Compute(m',q')$. 

		\item If $\mathcal{H}(T)$ is present in the Lookup Table $\mathcal{L}$, then get the AES encrypted instance of \~{m} from the db-store and add it to the result set $Res$.
		
	\end{enumerate}
	CS returns the result set $Res$ to the QU as the answer to the Hyper sphere range query.
	
	\item QU decrypts the elements of result set $Res$ and carries out the following validation step to filter out the false positives:
	\begin{itemize}
		\item Compute the distance between the decrypted data point and the query center.
		\item If the distance is greater than the query radius then the decrypted data point is discarded.
	\end{itemize}
	The validation step is needed because  $\mathcal{L}$ stores hash of the values which can lead to false positives. For standard hash function, the probability of such collisions is very less and thus very less false positives. In out experiments, we didn't encounter any false positives due to usage of hash in $\mathcal{L}$
\end{enumerate}

\subsubsection{Limitation}\label{prot1-limitation}
The query radius in $SHRQ_T$ is limited by the size $v$ of the Lookup Table. Specifically this protocol can only handle queries having $radius \leq v$. This is reflected in the Step 1 of Query Phase. This protocol is suitable for those applications where the maximum query radius is limited. Of course, the maximum query radius can be increased by increasing the size of Lookup Table. Generation of Lookup Table is a one time setup cost which can be amortized over large number of queries. However it does require extra storage at the cloud server for storing the Lookup Table.


\subsection{Coarse Grained Granular Storage (CGGS)}\label{sec:cggs}
In this section, we will describe a storage scheme, \emph{Coarse Grained Granular Storage (CGGS)} which helps in reducing the limitation of $SHRQ_T$.
The main idea of CGGS is :

\emph{Store multiple copies of the data item. In the first copy, actual data is considered for storage. In subsequent copies, the data is stored after transformation to a coarser euclidean space. In the coarser euclidean space, the unit distance between the two points is less than the distance  in the actual euclidean space. This reduction in the distance increases the  query radius range that can be supported given a Lookup Table.}

To formalize the coarser grained euclidean space we define the term :\\
\emph{f-coarser Euclidean Space ($f-c\mathcal{ES}$):} We define $f-c\mathcal{ES}$ as a transformed euclidean space where metric value $x$ in any dimension is transformed to $\lfloor \frac{x}{f} \rfloor$. In this notation, base euclidean space $\mathcal{ES}$ can also be represented as $1-c\mathcal{ES}$.

Here, $f$ is defined as the \emph{coarsity factor} which can take any positive value. Anyhow for the case we deal in this paper, $f$ can also be represented in the form of $f=b_c^{e_{\tiny{c}}}$, where $b_c$ represents \emph{coarsity base} while $e_c$ represents \emph{coarsity exponent}.
\vspace{1mm}
\\
\vspace{1.5mm}
{\Large\textbf{Distance Impact}}\\
Now we will quantify the relationship of the distances, between the points in the $1-c\mathcal{ES}$ and $f-c\mathcal{ES}$.
Let distance between two \emph{d-dimensional} points in base euclidean space ($1-c\mathcal{ES}$) be $dist_{1c}$ and distance between the same points in $f$-coarser transformed euclidean space be $dist_{fc}$. Then the following equations shows the relationship between the two values:
{\small
\begin{equation}\label{dist}
\frac{dist_{1c}}{f} - \sqrt{d} \leq dist_{fc} \leq \frac{dist_{1c}}{f} + \sqrt{d}
\end{equation}
}
{\Large\emph{Proof:} }
We prove equation \ref{dist} for 2-dimensional points. 
For higher dimensions, proof can be extended straight forwardly.
Deriving from \ref{dist}, it is sufficient to prove the following:
{\small
\begin{equation*}
(f.dist_{fc} - dist_{1c})^2 \leq 2f^2
\end{equation*}
}
Let us consider two 2-dimensional points $(x_1,y_1)$ and  $(x_2,y_2)$.
{\tiny
\begin{align*}
dist_{1c} &= \sqrt{(x_1-x_2)^2 + (y_1-y_2)^2}\\
dist_{fc} &= \sqrt{(\lfloor\frac{x_1}{f}\rfloor-\lfloor\frac{x_2}{f}\rfloor)^2 + (\lfloor\frac{y_1}{f}\rfloor-\lfloor\frac{y_2}{f}\rfloor)^2}\\
&\leq \sqrt{(\frac{x_1-x_2}{f}+1)^2 + (\frac{y_1-y_2}{f}+1)^2}
\end{align*}
}
Let $x_1-x_2 = X$ and $y_1-y_2 = Y$
{\tiny
\begin{align*}
(f.&dist_{fc} - dist_{1c})^2 \leq ({\sqrt{(X+f)^2+(Y+f)^2}-\sqrt{X^2+Y^2}})^2\\
&= 2(f^2 + X^2+Y^2+X.f+Y.f) - 2\sqrt{(X^2+Y^2+X.f+Y.f)^2 + f^2(X-Y)^2}\\
&\leq 2(f^2 + X^2+Y^2+X.f+Y.f) - 2\sqrt{(X^2+Y^2+X.f+Y.f)^2 }\\
&\leq 2f^2
\end{align*}
}

Equation \ref{dist} is used to resize the query radius when using data values stored in $x-c\mathcal{ES}'s$ where $x \geq 2$.
\subsection{Protocol SHRQ\texorpdfstring{\textsubscript{C}}{\texttwoinferior}}\label{protocol2}
In this section, protocol $SHRQ_C$ is presented which uses the storage scheme of \emph{CGGS} to overcome the limitations of $SHRQ_T$. In the protocol, DO stores multiple encrypted coarse grained copies of data at the CS (in addition to actual data). At the query time, the database copy of data with minimum coarsity ($f$) that  satisfies the query radius ($\leq \sqrt{v}$) is chosen for query execution.
For ease of explanation we chose \emph{coarsity base}, $b_c$ as $2$. We also assume that DO chooses a parameter $E_{max}$, representing  maximum coarsity exponent $e_c$.

\subsubsection{Setup Phase}
In this phase, DO creates multiple copies of data, encrypts them and uploads them to the cloud.
\begin{enumerate}
	\item DO invokes Keygen of $\mathcal{E}_{CES}$. 
	\begin{itemize}
		\item Parameters $\{N, \mathbb{G, G}_{T},e\}$ are shared with the Cloud Server. Also, DO and CS jointly agree upon on $\mathcal{H}$.
		\item Other parameters are kept secret by the DO.
	\end{itemize}
	
	\item DO encrypts each data point in the database as follows:-
	\begin{itemize}
		\item For $e_c$ = 0 to $E_{max}$
		\begin{itemize}
			\item Assign a unique id, $\mathring{i}$ to the data point.
			\item Compute the transformation of data point in $2^{e_{c}}-c\mathcal{ES}$ and then create the \emph{data component} vector. Let it be $m$.
			\item Use $TupleEncryption$  function of $\mathcal{E}_{CES}$ to encrypt $m$. The output is encrypted data point $m'$.
			\item Store ($\mathring{i}, m'$) at the CS in database instance $db-query^{e_c}$. DO creates separate database for each value of $e_c$.
		\end{itemize}
	
	\item Encrypt the data point with a separate encryption scheme like \emph{AES}. Let \~{m} represents the encryption of data point with \emph{AES}. 
		
	\item Store ($\mathring{i}$, \~{m}) at the CS in database \emph{db-store}.
		
	\end{itemize}
	
	\item DO calls the procedure \emph{CreateLookupTable}$(v)$ to generate $v$ sized lookup table ($\mathcal{L}$) and uploads it to CS.
	

\end{enumerate}

\subsubsection{Query Phase}
QU steps for query, $Q\{q,r\}$ are as follows:
\begin{enumerate}
	\item If $(r / 2^{E_{max}} + \sqrt{d}) > \sqrt{v}$ then the query is not supported by the current instantiation. No further processing is done. 
	
	\item QU computes  $e$ such that :
	\begin{enumerate}
		\item If $r  \leq  \sqrt{v}$, then $e=0$
		\item Else, $e$ is minimum $e_0$ s.t. $(r/2^{e_0} +\sqrt{d}) \leq  \sqrt{v}$, where $e_0 \in \{1,2,\cdots,E_{max}\}$
	\end{enumerate}

	\item Transform the query $Q\{q,r\}$ to $2^e-c\mathcal{ES}$, let the transformed query be $Q\{\hat{q},\hat{r}\}$. Transformed Query radius i.e. $\hat{r}$ is computed as following (it uses equation \ref{dist}) :
	\begin{enumerate}
		\item If $e=0$, $\hat{r} = r$
		\item Else, $\hat{r} = \lceil \frac{r}{2^e}+\sqrt{d}\rceil $
	\end{enumerate}
	The transformed query center i.e. $\hat{q}$ is computed by finding the transformation of $q$ in $2^e-c\mathcal{ES}$
	
	
	\item QU encrypts the query component w.r.t. $Q\{\hat{q},\hat{r}\}$ using $QueryEncryption$ function to get encrypted $q'$. QU then sends $q'$ to  CS for  further processing.
	\item CS performs similar operations as done in $SHRQ_T$. It starts with an empty result set \emph{Res} and for each  encrypted data point $m'$ in $db-query^{e}$, CS performs the following steps:
		\begin{enumerate}[i]
		\item Call \emph{Compute} function with  $m'$ and  $q'$. Let $T = Compute(m',q')$.
		\item If $\mathcal{H}(T)$ is present in the $\mathcal{L}$, then get the AES encrypted instance \~{m} from the db-store and add it to the result set $Res$.
	\end{enumerate}
	CS returns the result set $Res$ to the QU as its answer to the Hypersphere Range Query.
%
	
	\item Query User decrypts the elements of result set $Res$ using AES. Then it carries out the following validation step to filter out the false positives:
	\begin{itemize}
		\item Compute the distance between the decrypted data point and the query center.
		\item If the distance is greater than the query radius then the decrypted data point is discarded.
	\end{itemize}
	The validation step is needed because there will be some false positives due to: (1). Adjustment of $\sqrt{d}$ (due to equation \ref{dist})) made in radius value due to coarse grained execution, (2). Use of hash function in $\mathcal{L}$.
	As pointed in protocol $SHRQ_T$, the impact on false positives due to use of hash function in $\mathcal{L}$ is minimal.
\end{enumerate}

\subsubsection{Limitation} 
The result set returned by CS to the QU contains false positives.
 This is primarily because when executing in coarser grained (higher $c\mathcal{ES}$) the points that lie outside the circumference in $1-c\mathcal{ES}$ could now fall inside the hyper sphere due to transformation. Apart from this we also resized the radius(adding $\sqrt{d}$) to avoid false negatives due to equation \ref{dist}. Both the reasons collectively leads to false positives in the solution.

\subsection{Protocol SHRQ\texorpdfstring{\textsubscript{L}}{\texttwoinferior}}\label{protocol3}
In this section we will present protocol $SHRQ_L$, which improves upon the $SHRQ_C$. Specifically it reduces the storage space in comparison to $SHRQ_C$ and also does not give any \emph{false positives}.
However it does take more time in comparison to $SHRQ_C$. While $SHRQ_C$ works in a single iteration over the database, this protocol requires multiple passes over the database. Hence  provides an interesting design choice for the system admin.


This protocol uses the new designed \emph{`layered approach'} for query execution which helps in no false positives. In the first layer, the query is executed in $1-c\mathcal{ES}$. If the query radius $r > \sqrt{v}$, the result set will contain all the points that exist inside the hyperspherical disk bounded by radius width $[r-\sqrt{v},r]$. For the next layer, the query is executed in a higher $c\mathcal{ES}$. Radius chosen for the next layer of execution is $r-\sqrt{v}+f.\sqrt{d}$. This layer of execution will return points which exist inside a hyperspherical disk whose radius width is greater than the radius width of disk in previous layer (this is because the execution is happening in a coarser $\mathcal{ES}$). During the current layer's execution, the previous layers disk space acts as a buffer, so that points which fall inside the disk due to the distance impact (equation \ref{dist}) are true positive and not false positives.

The coarsity base $b_c$ for this protocol is carefully chosen by the following formula:
\begin{equation}\label{eq-bc}
b_c = \lfloor \frac{\sqrt{v}}{2\sqrt{d}+1}\rfloor
\end{equation}
\textbf{Explanation of equation \ref{eq-bc}:} The buffer space units (previous layer query space) in current  $c\mathcal{ES}$ is of the size of lookup table units supported by previous layer i.e $\sqrt{v}$. Buffer space units are $2\sqrt{d}+1$ due to the following:
\begin{enumerate}
	\item In the subsequent higher $c-\mathcal{ES}$, the radius value is adjusted by $+\sqrt{d}$ to adjust for the points falling outside the boundary of hypersphere.
	\item Another $\sqrt{d}$ addition is to consider the points that fall inside the hemisphere of adjusted radius.
	\item The additional plus 1 is added to adjust the ceiling operation that is used in radius computation for higher   $c-\mathcal{ES}$.
\end{enumerate}

The value of coarsity base for this protocol helps in reducing the storage requirements as compared to $SHRQ_C$.
In $SHRQ_L$, for supporting a given  query radius range $[0,R_{max}]$ less number of db stores ($\log_{b_c}R_{max}$) are required as compared to $SHRQ_C$, ($\log_{2}R_{max}$). \\

The details of $SHRQ_L$ as follow:
Like before, DO choses the parameter $E_{max}$ representing the maximum coarsity exponent.


\subsubsection{Setup Phase}
In this phase, DO creates multiple copies of data, encrypts them and uploads them to the cloud.
\begin{enumerate}
	\item DO invokes Keygen of $\mathcal{E}_{CES}$. 
	\begin{itemize}
		\item Parameters $\{N, \mathbb{G, G}_{T},e\}$ are shared with the Cloud Server. Also DO and CS jointly agree on Hash function $\mathcal{H}$.
		\item Other parameters are kept secret by the DO.
	\end{itemize}
	
	\item DO encrypts each data point in  database as follows :
	\begin{itemize}
		\item For $e_c$ = 0 to $E_{max}$
		\begin{itemize}
			\item Assign a unique id, $\mathring{i}$ to the data point
			\item Compute  transformation of data point in  $b_{c}^{e_{c}}-c\mathcal{ES}$ and then create the \emph{data component} vector. Let it be $m$.
			\item Use $TupleEncryption$  function of $\mathcal{E}_{CES}$ to encrypt $m$. The output is encrypted data point $m'$.
			\item Store ($\mathring{i}, m'$) at the cloud server in database $db-query^{e_c}$. Data Owner creates separate database for each value of $e_c$.
		\end{itemize}
		
		\item Also, encrypt the data point with a separate encryption scheme like \emph{AES}. Let \~{m} represents the encryption of data point with \emph{AES}. 
		
		\item Store ($\mathring{i}$, \~{m}) at the cloud server in database \emph{db-store}.
	\end{itemize}
	
	\item DO calls the procedure \emph{CreateLookupTable}$(v)$ to generate $v$ sized lookup table ($\mathcal{L}$) and uploads it to the CS.

\end{enumerate}

\subsubsection{Query Phase}
QU initiates query phase agiainst the query, $Q\{q,r\}$ and steps are as follows:
\begin{enumerate}

	\item If $(r / 2^{E_{max}} + \sqrt{d}) > \sqrt{v}$ then the query is not supported by the current instantiation. No further processing is done. 
	\item QU computes the query radius values that are used in subsequent layered execution using  Algorithm \ref{layered_radius}.
	\begin{algorithm} 
		\DontPrintSemicolon
		
		\KwIn{$r$, $\sqrt{v},d,b_c$}
		\KwResult{\emph{Map m} : radius values in subsequent layers}
		\textbf{Inititialize} : $Map<Layer,Radius>,map=\psi$\\
		\emph{Add Radius w.r.t. Layer 0, i.e. $\{0,r\}$ to map} \\
		\emph{Set }$r=r-\sqrt{v}$\emph{ and Set $i=1$}\\
		\While{$r>0$ }{
			$r=r + (b_c)^i.\sqrt{d}$\;
			\emph{Add $\{i,r\}$ to map}\;
			$r=r-(b_c)^i.\sqrt{v}$\;
			$i=i+1$\;
		}
		\caption{\textbf{Layered Radius Values Generation Procedure}}
		\label{layered_radius}
	\end{algorithm}
\end{enumerate}
For each \emph{layer i} $\in map$ repeat following steps 3-5 with $r = map.get(i)$. Initialize empty result set $Res$.
\begin{enumerate}
	\setcounter{enumi}{2}

	\item Transform the Query $Q\{q,r\}$ in $(b_c)^i-c\mathcal{ES}$ to get $Q\{\hat{q},\hat{r}\}$. Transformed Query radius $\hat{r}$ is computed as following (using equation \ref{dist}) :
	\begin{enumerate}
	\item If $i=0$, $\hat{r} = r$
		\item Else, $\hat{r} = \lceil \frac{r}{(b_c)^i}\rceil $
	\end{enumerate}
	The transformed Query center i.e. $\hat{q}$ is computed by finding the transformation of $q$ in $(b_c)^i-c\mathcal{ES}$.

	\item QU encrypts the query component w.r.t. $Q\{\hat{q},\hat{r}\}$ using $QueryEncryption$ function to get encrypted $q'$. QU then sends $q'$ to  CS for  further processing.
	\item CS performs operations similar to done in protocols $SHRQ_T, SHRQ_C$. For each of the encrypted data point $m'$ in $db-query^{i}$, CS performs the following steps:
	\begin{enumerate}[i]
	\item Call \emph{Compute} function with  $m'$ and  $q'$. Let $T = Compute(m',q')$.
		\item If $\mathcal{H}(T)$ is present in the $\mathcal{L}$, then get the AES encrypted instance \~{m} from the db-store and add it to the result set $Res$.
		
	\end{enumerate}
	\item CS returns the result set $Res$ to the QU as its answer to the Hypersphere Range Query.
	
	\item QU decrypts the elements of result set $Res$ using AES. Thereafter it carries out the following validation step to filter out the false positives:
	\begin{itemize}
		\item Compute the distance between the decrypted data point and the query center.
		\item If the distance is greater than the query radius then the decrypted data point is discarded.
	\end{itemize}
	At the end of the validation step, Query User has the plain text result set for his Hypersphere range query.
	The validation step is needed for the similar reason as in protocol $SHRQ_T$. 
\end{enumerate}
\vspace{2mm}
{\Large\textbf{Database Updates}}
\vspace{1mm}\\
The proposed system  supports  dynamic insertions/updates. Any new data point can be uploaded by DO to CS using data point encryption (\emph{Step 2 of Setup Phase of all the protocols}). However, for dynamic updates DO first removes the existing  records and then reinserts the updated record. Our system stores the AES encrypted values w.r.t. each data point (used for referencing of records) and index $\mathring{i}$ is referenced with every storage of $\mathcal{E}_{CES}$ encrypted and AES encrypted data values. The updating record can be reference and removed by using AES encryption and the respective index $\mathring{i}$.
\section{Secure Range Query}
In this section, we adapt the \emph{SHRQ} protocols to build a solution for the problem of \emph{ Secure Range Query}.

\subsubsection*{Secure Range Query (SRQ)} SRQ is defined as following:\\
\emph{Given an encrypted database $D'$, 
	Cloud Server should be able to correctly retrieve the records which satisfy the range predicate in the encrypted queries $q'$ such that CS doesn't knows the underlying plain text information of $D'$ and $q'$. } 

Range Query over the column can be segregated by the following types:
\begin{enumerate}[i.]
	\item Closed Range : These queries involve selection of values between close range $[V_{Left},V_{Right}]$. For example, for a query \emph{``select * from employees where age $\geq$ 25 and age $\leq$ 50''}, close range over \emph{``Age''} column is $[25,50]$.
	\item Open Range: These queries involve selection of values w.r.t. open range $[V_{Left},\infty)$ or $(-\infty,V_{Right}]$. For example, in \emph{``select * from employees where age $\geq$ 25'' }, open range over \emph{``Age''} column is $[25,\infty]$.  
\end{enumerate}
It can be observed that given the statistical information of \emph{(min, max)} w.r.t. column, every open range query can be converted to close range query. For example, if it is known that $max(Age) = 60$, then the open range query $[25,\infty]$ over \emph{``Age''} column can become a close range query $[25,60]$. Both the queries will retrieve the same records.

To solve \emph{SRQ} using \emph{SHRQ}, it can be observed that \emph{SHRQ} solution w.r.t. single dimension data retrieves records in the close range $[center-Radius,center+Radius]$. Therefore given a close range query over a column $[V_{Left},V_{Right}]$, it can be modeled as a \emph{SHRQ} query in single dimension with \emph{query point} $ = \frac{V_{Left}+V_{Right}}{2}$ and $Radius = \frac{V_{Right}-V_{Left}}{2}$.

To incorporate the support for \emph{SRQ} queries, data owner has to consider the change in Data Component and Query Component vectors as below:
	\begin{itemize}
	\item Data Component : Unlike before (section \ref{objective}), the Data Component vector will now be a $(2d+1)$ dimensional vector: $\{m_1,\cdots,m_d,1,m_1^2\cdots,m_d^2\}$
	\item Query Component: Query Component vector will now consider the type of the query and will format the Query Component vector as following:
	\begin{enumerate}
		\item \emph{SRQ Query:} Assuming the range query $[V_L,V_R]$ is asked over column index $i$, then the equivalent \emph{HRQ} query $Q\{q_i,r\}$, correspondingly is represented as $Q\{\frac{V_L + V_R}{2},\frac{V_R - V_L}{2}\}$. The data component representation of the same query is as below:
		\begin{equation*}
		\{0_1,\cdots, 2q_i,\\\cdots,0_d,R^2-||q_i||^2,0_1,\,-1_i,\cdots,0_d\}
		\end{equation*}
		\item \emph{SHRQ Query:} Query component w.r.t. \emph{d dimensional SHRQ query} is as below:
		\begin{equation*}
		\{2q_1,2q_2,\cdots ,2q_d,R^2-\sum_{i=1}^{d}q_i^2,-1_1,\cdots,-1_d\}
		\end{equation*}
		\emph{SHRQ} query over arbitrary columns can be modeled by modifying above expression keeping only the required column terms and replacing rest of the terms by $0$.
	\end{enumerate}
	\end{itemize}
Apart from the Component vectors change another change to be made is in \emph{KeyGen} function of $\mathcal{E}_{CES}$ where the size of vectors $A,B$ would now be $2.d+1$. Rest of the Scheme and protocol remains same.

It can be observed that the size of the Query Component remains the same on whether the query is \emph{SRQ} or \emph{SHRQ}. And since the sizes are same, the CS will  be oblivious to the type from the encrypted query $q'$. In case of \emph{SRQ}, he will not be able to identify even the column on which the range predicate is applied. Thus making the scheme \emph{Type and Column Oblivious}.

In contrast to the existing solutions of \emph{Order Preserving Encryption (OPE)}, our solution doesn't leak any information w.r.t. ordering of column data. And the oblivious nature of our scheme hinders the adversary from gradual building any solution maps using the subsequent query execution output. Also the current state of the art $OPE$ suffers from the stateful nature of the scheme, where the \emph{client-server} collectively manages the dynamic state and the encryption of a new value requires mutation of the subset of already stored encrypted values. While the proposed scheme is a stateless scheme supporting any dynamic updates or insertions.

\section{Security} \label{security}
We present the security of our protocols in two attack models: \emph{COA} and \emph{KPA}

\subsection*{Ciphertext Only Attack (COA)}
In COA attack model, adversary has access to only the ciphertext values. He can also observe the computation tasks performed over the encrypted data. His aim is to decipher the encrypted database points. 

In the presented scheme, the COA adversary can observe the computational output and lookup table access in the query phase execution. Any reverse map building using the adversarial guess on the computation value will first require adversary to identify the type of  query. Our protocol can be used to solve:
\begin{enumerate}
	\item  Range query over encrypted database column (column oblivious). 
	\item Hypersphere range query over arbitrary number of columns. 
\end{enumerate}
For both the solutions, the protocol execution in the query phase is similar. Apart from guessing the type, the adversary should guess the column imprints over which the query executes.
This variable (type of query) consideration with every query execution will make it difficult for COA adversary to build a useful information map and thus will make it difficult to decipher the encrypted database points.

\subsection*{Known Plaintext Attack Model (KPA)} 
In this attack model, adversary also has access to few plaintext and corresponding ciphertext values. The goal of  adversary is to know the plain text of rest of the accessible encrypted database.

KPA security analysis w.r.t. spatial range queries, like kNN, range query etc. involve the study of attack feasibility  to reach KPA adversaries goal. The attack metric used is  in terms of the amount of  \emph{Plaintext (PT)} and \emph{Ciphertext(CT)} values required to decrypt the whole database.

To break a single  \emph{d dimensional} encrypted data point $DP$, the attack construction consists of building a solvable set of \emph{d linear equations}. The linear system is usually built by using the \emph{PT, CT} pairs that are available to adversary under KPA. For the attack on $DP$, he uses access of  $d$ such pairs of Query Points $QP_i$, $1\leq i \leq d$ and the scheme output of each $QP_i$ with $DP$. Following we show such a system of equation.
\begin{equation*}
\begin{aligned}
\begin{bmatrix} \text{\emph{d PT}}\\\text{\emph{DP}}\\ {varaibles} \end{bmatrix}_{1\times d} \times 
&\begin{bmatrix}
\begin{bmatrix} \text{\emph{PT of}}\\QP_1 \end{bmatrix}^T
\begin{bmatrix} \text{\emph{PT of}}\\QP_2 \end{bmatrix}^T
\cdots
 \begin{bmatrix} \text{\emph{PT of}}\\QP_d \end{bmatrix}^T
\end{bmatrix}_{d\times d}\\
&= 
\begin{bmatrix} 
\text{\emph{Scheme Output}} \\ 
\text{\emph{given CT's of}} \\ 
\text{\emph{DP and QP}}_i  
\end{bmatrix}_{1 \times d}
\end{aligned}
\end{equation*}
The adversary would reach a solvable system of equations, if he gets the useful inferable information from the scheme output, thus deciphering the particular data point $DP$. For decrypting the whole database, it reduces to the fact on whether the same set of pairs ($PT,CT$ pairs for $QP_i$'s) could be used in breaking  other data points or not.

We present the KPA security analysis in the same attack setting to decrypt the whole encrypted Database.

In our schemes presented in section \ref{protocols}, the CS side computation involving encrypted DP and encrypted QP consists of computation of deterministic modular exponent value $(\text{\emph{in }}  \mathbb{G}_T)$ and verification of same value in the lookup table. W.r.t. deterministic value, the information leakage requires adversary to have high computation power (\emph{Secure under hardness of Discrete logarithm problem}). The leakage w.r.t. lookup table access  (\emph{for True Positives})  leaks the  following information:
\begin{itemize}
	\item For Query phase execution in $1-c\mathcal{ES}$, the adversary  can guess the scheme computation value i.e. $\text{\emph{Radius}}^2 -\text{\emph{Dist(DP,QP)}}^2$ in range $[0,len(\text{\emph{lookup table}})]$.
	\item For Query phase execution in higher coarser spaces, $f-c\mathcal{ES}$, the adversary guesses the scheme computation value i.e. $\text{\emph{Radius}}^2 -\text{\emph{Dist(DP,QP)}}^2$ in range $[0,f\times len(\text{\emph{lookup table}})]$. 
\end{itemize}
For adversary, to decipher a particular encrypted data point$-DP'$ by building the system of linear equations following conditions needs to meet:
\begin{enumerate}
	\item The adversary access to \emph{d Plaintext-Ciphertext Query points} s.t. $DP'\in \text{\emph{True positives}}$ for each of the query.
	\item Adversary correctly guesses the computation output w.r.t. lookup table access.
\end{enumerate}
In KPA attack model, adversary doesn't have the choice on \emph{PT-CT Query points} that are revealed to him. So, to break the particular point, number of \emph{PT-CT} pairs could range between $[d-|D|]$.

And, to decipher the whole database assuming the metric range in each dimension being $X$ and lookup table size $v$, the number of queries could range between $[d.(\frac{X^d}{v^d}),|D|]$. which would be of the order of database size itself.

\section{Experiments}\label{experiments}

In this section, we present the experimental evaluation of our proposed Protocols. The experiments were carried out on two machines, one acting as the DO and the other  as  CS. Both the machines have the following configuration, Ubuntu 16.04.5 LTS running on 16 core Intel(R) Xeon(R) Gold 6140 CPU @ 2.30GHz and 64 GB memory. The protocols were written in Java using Java Pairing-Based Cryptography Library (JPBC) \cite{moreJPBC} using the group size of 1024 bits. The two machines were connected to each other via 
high bandwidth network.

Figure \ref{fig:data_enc_dim} shows how data encryption per database copy varies with the number of dimensions in the data. From the figure, it is clear that data encryption time increases linearly with the number of dimensions. It is expected since work to be done increases linearly with number of dimensions in the data.
It is important to note here that the data encryption time per copy is same for all the proposed Protocols. \emph{$SHRQ_T$} has only one copy of data, so this figure represents data encryption time for \emph{$SHRQ_T$}. On the other hand, \emph{$SHRQ_C$} and \emph{$SHRQ_L$} have multiple copies of databases depending on configuration. Hence the total data encryption time for them will be  multiple of the numbers presented in the figure.

Figure \ref{fig:data_enc_dp} shows how data encryption per database copy varies with number of data points. The figure clearly shows an expected linear relationship between the number of data points and time taken to encrypt the database copy. 
It is important to note here that data encryption time per copy is same for all the proposed Protocols. \emph{$SHRQ_T$} has only one copy of data, so this figure represents data encryption time for \emph{$SHRQ_T$}. On the other hand, \emph{$SHRQ_C$} and \emph{$SHRQ_L$} have multiple copies of data based on configuration. Hence the total data encryption time will be a multiple of the numbers presented in the figure.

Figure \ref{fig:prot_dim} shows how the query execution varies with the number of dimensions in the data. The figure shows a linear correlation between number of dimensions and query execution time. This is expected since increasing the number of dimensions causes increased effort to compute the distance between the data point and query point.
The figure also shows that both \emph{$SHRQ_T$} and \emph{$SHRQ_C$} have almost identical running times. This is because both of these protocols make a single pass over the data. \emph{$SHRQ_C$} has some extra false positives which have to removed at the end, leading to slightly increased running time.
On the other hand, \emph{$SHRQ_L$} requires much higher running time than both \emph{$SHRQ_T$} and \emph{$SHRQ_C$}. This is because \emph{$SHRQ_L$} has to make multiple pass over the data (two in our experiments). 

Figure \ref{fig:prot_dp} shows how the query execution varies with the number of data points. The figure shows an expected linear correlation between number of data points and query execution time.
The figure also shows that both \emph{$SHRQ_T$} and \emph{$SHRQ_C$} have almost identical running times. This is because both of these protocols make a single pass over the data. \emph{$SHRQ_C$} has some extra false positives which have to removed at the end, leading to slightly increased running time.
On the other hand, \emph{$SHRQ_L$} requires much higher running time that both \emph{$SHRQ_T$} and \emph{$SHRQ_C$}. This is because \emph{$SHRQ_L$} has to make multiple pass over the data (two in our experiments). 

\begin{figure}
	\includegraphics[width=\linewidth]{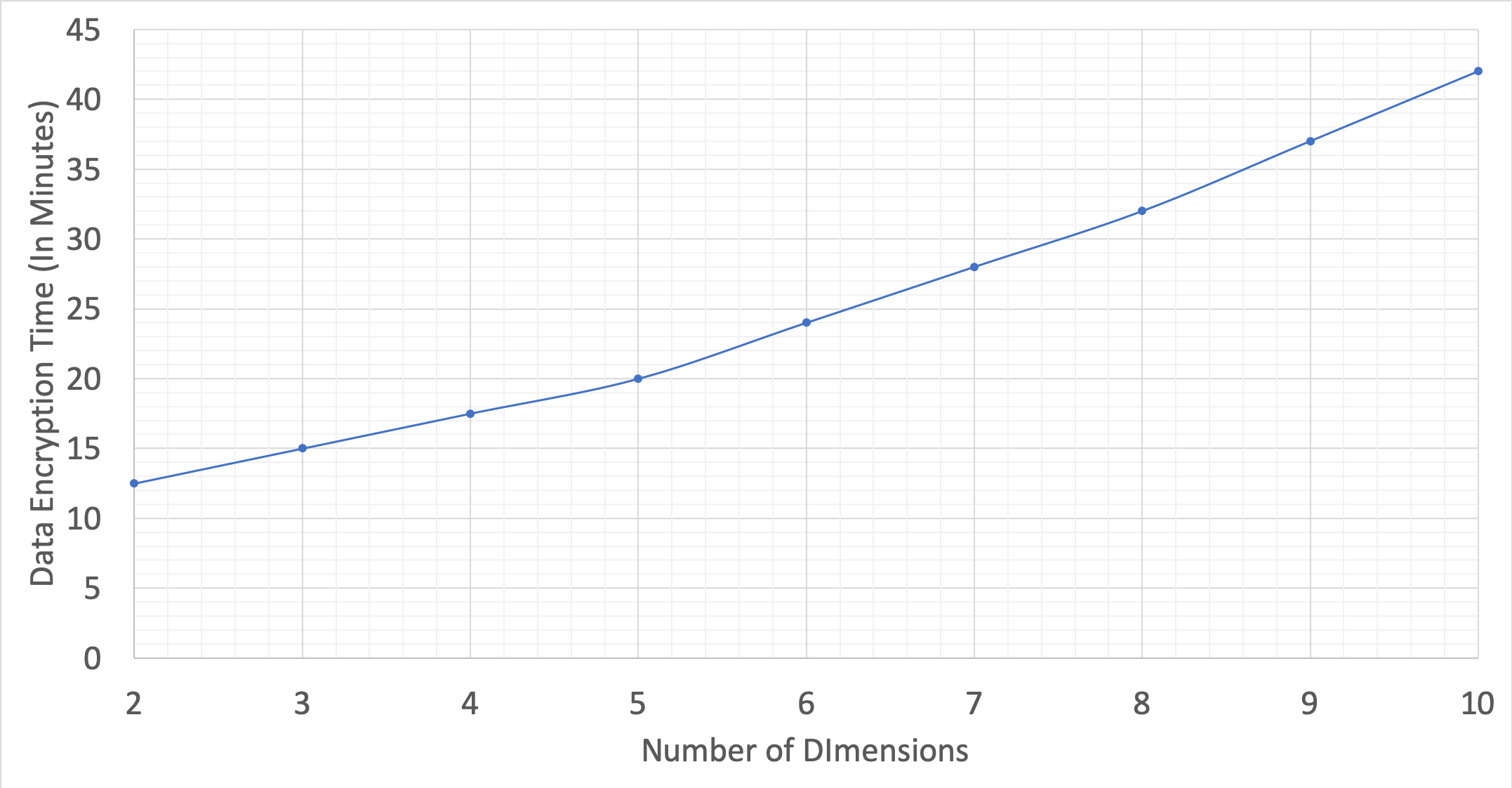}
	\caption{Impact of Data dimensions on Data Encryption per database copy (1000 Data Points)}
	\label{fig:data_enc_dim}
\end{figure}

\begin{figure}
	\includegraphics[width=\linewidth]{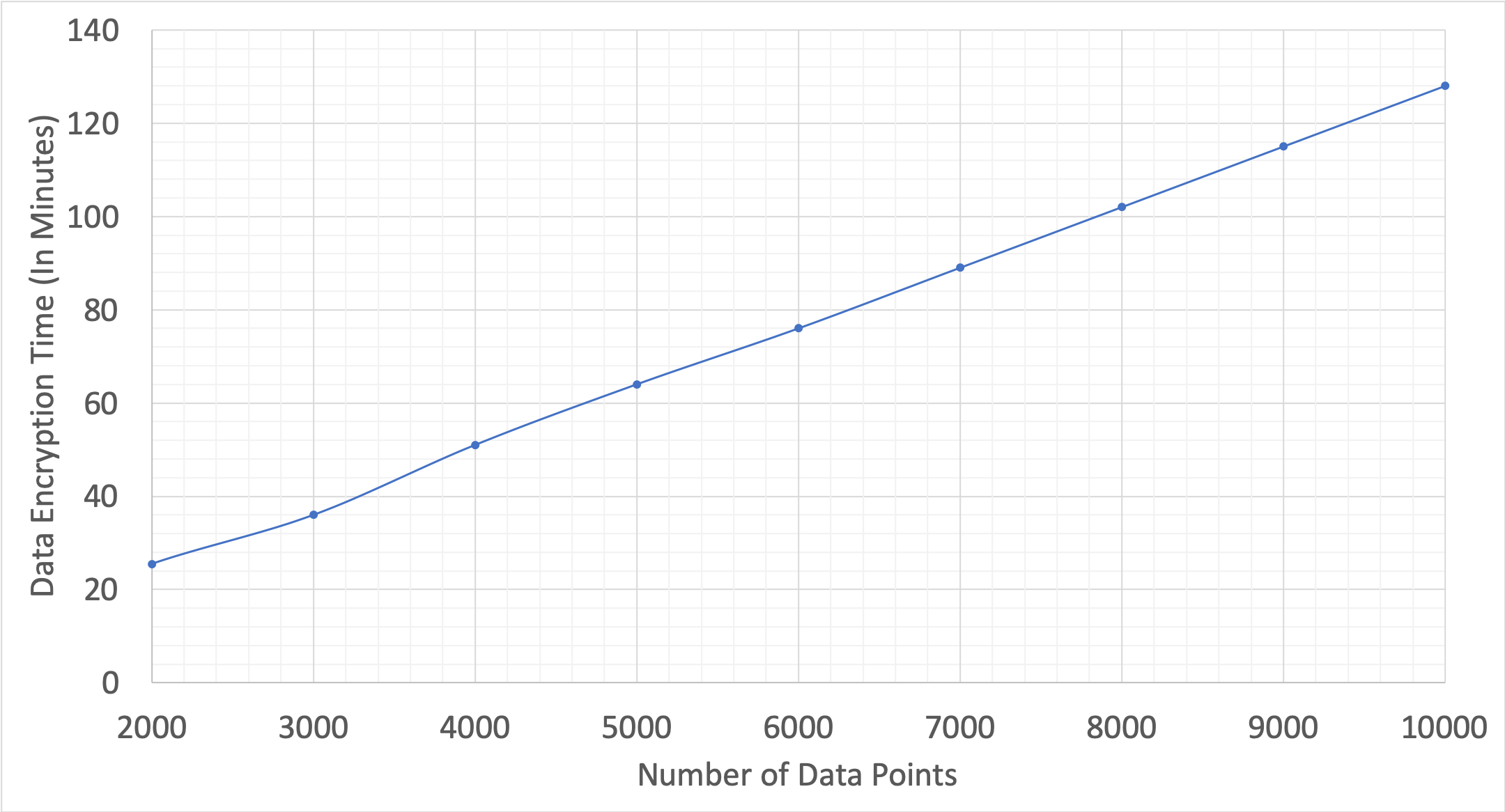}
	\caption{Impact of number of data points on Data Encryption per database copy (d=2)}
	\label{fig:data_enc_dp}
\end{figure}

\begin{figure}
	\includegraphics[width=\linewidth]{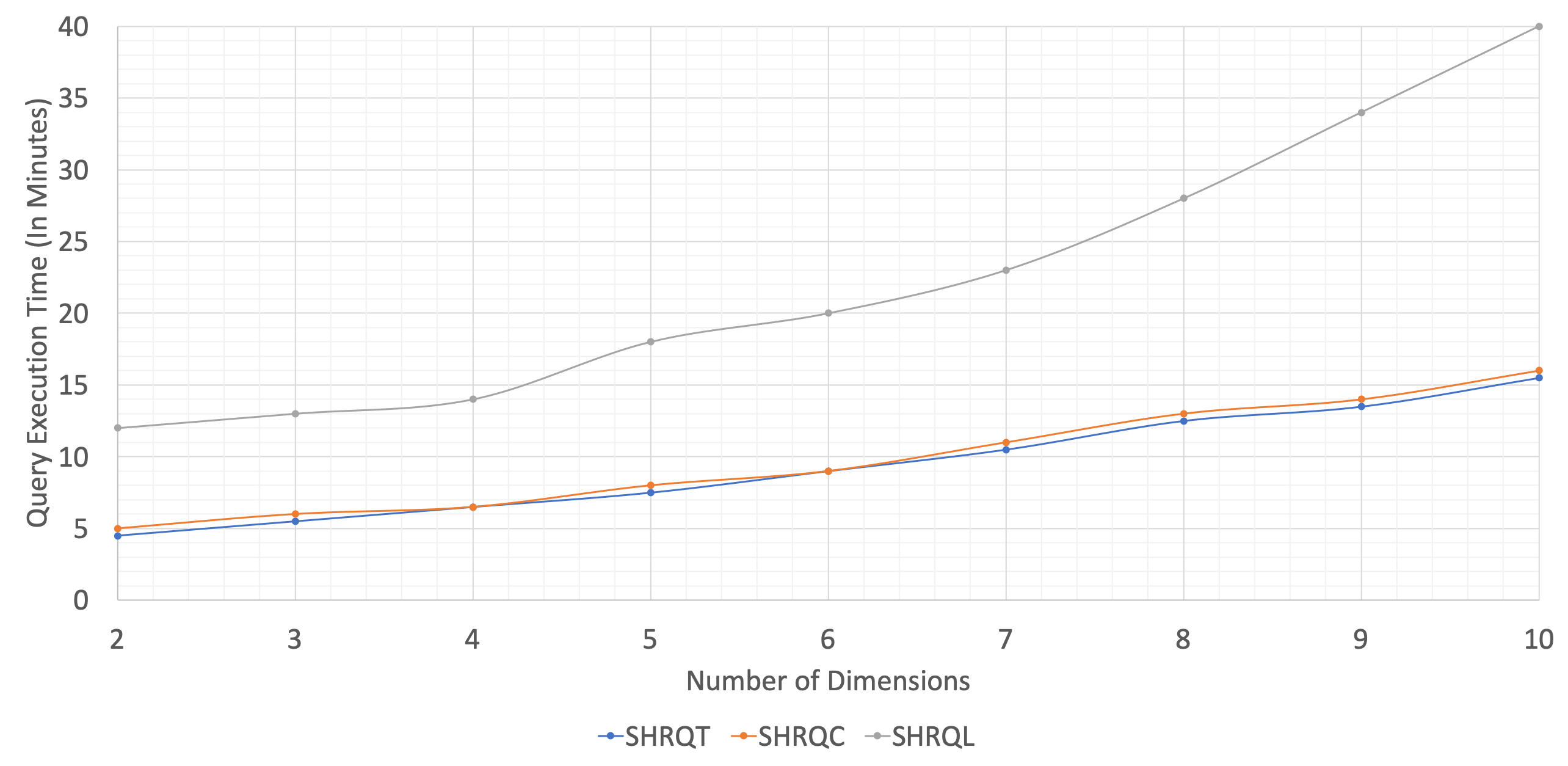}
	\caption{Impact of Data dimensions on Protocols (1000 Data Points)}
	\label{fig:prot_dim}
\end{figure}

\begin{figure}
	\includegraphics[width=\linewidth]{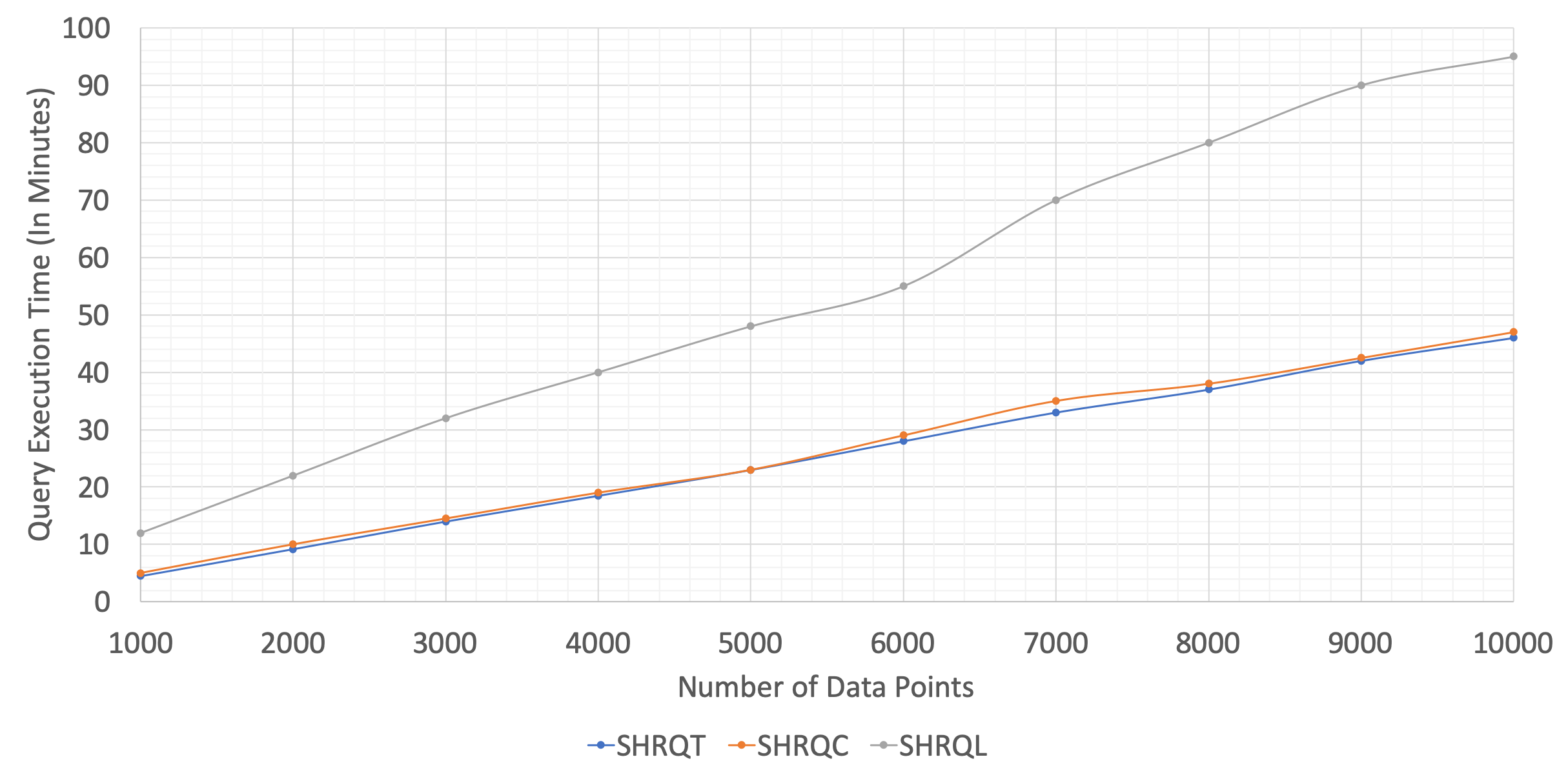}
	\caption{Impact of number of data points on Protocols (d=2)}
	\label{fig:prot_dp}
\end{figure}
\subsection*{Scope of Improvement}
The performance bottleneck in our protocols are the 
costly pairing operations (\emph{Homomorphic Multiplication}) between the encrypted query point and  data points in the big composite order group(1024 bits). The breakthrough work of Freeman \cite{more2010Freeman}, presents a framework where big composite order based Elliptic Curve (EC) schemes can be converted to comparably small prime order EC schemes. Using the framework they designed an modified version for BGN encryption scheme which improved the pairing operation by 33 times (Table 7 in \cite{more2013Guillevic}).
For our proposed scheme of $CES$, a separate design study is required for converting it to prime order EC scheme and has been left as future work.
 
\subsection*{Comparison to State-of-the-Art}
The closest relation of our work is  to two papers from Wang \emph{et al.} \cite{wang2015circular,wang2016geometric} which cover the problems of \emph{Circular Range Query(CRQ)} and \emph{Geometric Range Queries(GRQ)} over encrypted data respectively. Both the solutions suffer  slowness in their \emph{QueryEncryption} algorithms, which is directly proportional to solution set range. Solution of \cite{wang2015circular}, also suffers slowness in the \emph{Query Execution} over the CS. For radius values of 1 unit, 5 units and 10 units, the execution time over CS for a $|D|=1000$ and $d=2$, is around 10 secs, 40 secs and 100 secs. Thus showing linear rise with the rise in radius. Considering only the \emph{Query Execution} over CS, \cite{wang2016geometric} performs comparable to our system with time of 800 secs when $|D|=1000$ and $d=2$. However, altogether including the QueryEncryption times our scheme outperforms both of them.

The other relation of our work can be made to the papers dealing with \emph{Secure k Nearest Neighbors (SkNN)} problem. Beginning with the \emph{SkNN} solution presented by Wong \emph{et al.} \cite{kNN2009Wong} which used the idea of invertible  matrix multiplication, the performance was comparable to plain text kNN execution. However, Yao \emph{et al.} \cite{kNN2013Yao} broke the KPA security claim of \cite{kNN2009Wong}. \cite{kNN2013Zhu} also presented a Vornoi diagram based \emph{SkNN} solution which was very efficient but only worked for cases when $d=2$. Another work from Elmehdwi \emph{et al.} \cite{kNN2014Elmehdwi}, proposed a \emph{SkNN} solution in 2-cloud model where one of the cloud is having the encryption key. They developed the secure multiparty protocols in the proposed setting to develop the  solution satisfying stronger security notion of \emph{CPA}. However, the solution did suffer from huge latency which was directly proportional to $k$, $|D|$ and the domain size of each attribute$(l)$. Specifically, for $|D|=2000, l=12,d=6$ and $k=25$, the scheme took 650 minutes for completion. In comparison, our scheme is independent of $l$ and under similar security parameters took 48 minutes to completion.
\section{Related Work}
There has been considerable amount of research in the area of secure query processing over the encrypted data. Although the ground breaking work of Gentry \cite{gentryFHE} on \emph{Fully Homomorphic Encryption(FHE)} enabled arbitrary polynomial computation over encrypted data, the scheme had a high computational cost. Thus making it impractical for real world applications. Many application specific \emph{Partial Homomorphic encryption} schemes have been developed. Secure solutions varying from \emph{Geometric Queries},
\emph{Nearest Neighbors search(kNN)},
 \emph{Range predicate} 
 etc. to  \emph{database level secure systems} 
 exists in prior literature.

\textbf{Geometric Queries:} The only known papers in this domain are from Wang \emph{et al}. \cite{wang2015circular,wang2016geometric}. We analysed the performances of these scheme in section \ref{experiments}. At the core, both the schemes use \emph{Predicate Evaluation}  scheme of Shen \emph{et al} \cite{more2009Shen} using which it can be \emph{securely} checked if dot product of two vectors is 0 or not. 
In \cite{wang2015circular} , CRQ problem is solved by finding points on circumference of the circle and repeating the task by varying radius of cicles. While in\cite{wang2016geometric}, GRQ problem is addressed where each data point is stored using separate bloom filter. The  query phase includes building the \emph{bloom filter} of all points in query space and then using the idea of \cite{more2009Shen} it is checked if data bloom filter state exist in query bloom filter. As discussed earlier, both the solutions suffer from huge slowness in query execution phase compared to our solution.

%
\textbf{Nearest Neighbors:} Alot of literature exist in solving the \emph{Nearest Neighbors} problem over the encrypted data. We have covered the subset of these solutions ( Wong \emph{et al.} \cite{kNN2009Wong},Yao \emph{et al.} \cite{kNN2013Yao} and Elmehdwi \emph{et al.} \cite{kNN2014Elmehdwi}), in section \ref{experiments}. 
Following we give brief of  other solutions in same problem domain.

Hu \emph{et al}. \cite{kNN2011Hu} and Wang \emph{et al}. \cite{kNN2016Wang} built the solutions by developing technique of secure traversal over encrypted R-Tree index structure. While the performance of \cite{kNN2016Wang} is better than \cite{kNN2011Hu}, both the schemes suffers by not supporting the dynamic update or insertion of new records. Any such updates require the change in the index structure, which would be a costly operation. In contrast, we support dynamic updates and insertion.


 PIR (Private Information retrieval) for secure k-NN has been studied by Ghinita \emph{et al}. \cite{kNN2008Ghinita}, Papadopoulos \emph{et al}. \cite{kNN2010Papadopoulos} and Choi \emph{et al}. \cite{kNN2014Choi}. PIR allows users to retrieve  data stored at the CS in oblivious manner, so that user's query point is not revealed and also server doesn't know about which records are accessed. The data stored at server in such schemes is kept in plain text format which is different form our setting of work.

Schemes by Zhu \emph{et al}. \cite{kNN2013Zhu,kNN2016Zhu} and Singh \emph{et al}. \cite{kNN2018Singh} presented the solutions in multi-user setting. They considered the \emph{Query Privacy} metric in their solutions, so that Query Users query is not known to Data owner and Cloud server. The underlying idea of the schemes is similar to matrix multiplication idea of \cite{kNN2009Wong}.

Lin \emph{et al}. \cite{kNN2017Lin} extended the attack study of \cite{kNN2013Yao} and presented attacks on the related schemes
Lei \emph{et al}. \cite{kNN2017Lei} proposed a SkNN solution for 2-dimensional points by using LSH (Location sensitive hashing). They first construct the secure index around the data and then outsource the data and index to the cloud. Since the scheme uses LSH data structure, their result contains false positives.
\\
\textbf{Range Predicates:} Considerable amount of research exist in this area. Solutions ranging from \emph{order preserving encryption(OPE)} schemes to schemes specifically evaluating range predicates exist. Agrawal \emph{et al}. \cite{ope2004Agrawal} developed the first order preserving encryption scheme. They mapped the known plain text distribution with randomly chosen distribution. The mapped values represented the OPE ciphertext of plaintext values.
Boldyreva \emph{et al}. \cite{ope2009Boldyreva,ope2011Boldyreva} introduce the security notions w.r.t. OPE problem.
 They also presented the \emph{OPE} scheme that uses hypergeometric distribution
 to generate ordered cipher text. 
 Popa \emph{et al}. \cite{ope2013Popa} presented the first ideally secure (indistinguishability under ordered chosen plaintext attack, IND-OCPA) OPE scheme. The scheme presented is stateful and requires multiple round communications to generate ciphertext.
Kerschbaum \emph{et al}. \cite{ope2014Kerschbaum} presented the more efficient IND-OCPA secure scheme. The number of round communications are reduced to constant number in their scheme. Kerschbaum \cite{ope2015Kerschbaum}  also presented  \emph{Frequency hiding OPE} scheme, where the ciphertext of the scheme is randomized in the range of values.

Hore \emph{et al}. \cite{range2004Hore,range2012Hore} presented a scheme that built privacy preserving indices over sensitive columns using data partitioning techniques.The built index helped in supporting the obfuscated range queries. 
Li \emph{et al}. \cite{range2014Li} presented a solution using specially designed index tree from data, where all the binary prefixes of the data are stored. Further they used \emph{Bloom Filters} to hide the prefix information. Any range query over data is then converted to similar prefix identification over index structure. Due to  bloom filters scheme suffers from \emph{false positives}. Karras \emph{et al}. \cite{range2016Karras} presented range predicate solution using Matrix multiplication based technique. However Horst \emph{et al}. \cite{range2017Horst} presented the cryptanalysis of the scheme where they were able to break their security claims.
\\
\textbf{Secure Database Systems:} Hacig\"{u}m\"{u}\c{s} et al \cite{database2001Hacigumus} proposed the first system for executing SQL queries over encrypted data. They used bucketization technique at server to approximate the filtering of result set and thereafter proposed clients to perform final query processing. The secure database systems: CryptDB presented by Tu \emph{et al}. \cite{database2011Popa} and Monomi presented by Tu \emph{et al}. \cite{database2013Tu} used multiple encryption
schemes to give support for multiple SQL predicates. They
used OPE, Paillier encryption \cite{more1999Paillier} and Text Search \cite{search2000Song} schemes to provide support for the range predicate, addition and search queries respectively. 
However, both the schemes fail to support complex queries
where multiple operators are involved in a single query. 

In schemes of Bajaj \emph{et al}. \cite{database2011Bajaj} and Arasu \emph{et al}. \cite{database2013Arasu,database2015Arasu}, cloud uses the secure hardware for secure computation where the encryption keys are stored. For complex computations, the data is first decrypted at the secure hardware, then processed for evaluation and then the answer is encrypted back and sent back to client. Though the systems performance is good, it still has to keep the encryption keys at third party location, which could lead to trust deficit.

SDB presented by \cite{database2014Wong} performed query processing with a set of secure data-interoperable operators by using asymmetric secret-sharing scheme. They provided protocols to handle queries having across the column computations. Hahn \emph{et al}. \cite{database2019Hahn} presented the technique for secure database joins. They used searchable symmetric encryption and attribute based encryption to develop protocol for secure join.


\section{Conclusion and Future Work}
To conclude,  we presented a novel solution to solve the problem of \emph{Secure Hyper Sphere Range Query(SHRQ)} over encrypted data. While  existing solutions suffered from  scalable slowness for big range queries ($\mathcal{O}(R^2)$), our scheme guarantees constant time performance ($\mathcal{O}(c)$). Along with SHRQ solution, we also presented the \emph{Secure Range Query(SRQ)} solution with proven KPA security and column obvious execution.

Future work directions on the proposed system includes: 1) Designing the solution using prime order EC curves. It is known that prime order EC curves perform better than composite order EC curves. 2) Extend the solution in framework where DO and QU are considered as separate entities. 3) Performing the analysis to support  other spatial queries like kNN, Geometric queries etc. 

\bibliographystyle{ACM-Reference-Format}
\bibliography{vldb_sample}

\end{document}